\documentclass[11pt,reqno]{amsproc}
\usepackage{amsaddr}
\usepackage[top=1in, bottom=1in, left=0.75in, right=0.75in]{geometry}

\usepackage{url,hyperref,lineno,microtype,subcaption}
\usepackage{xspace}

\usepackage[
backend=biber,
style=alphabetic,
sorting=ynt
]{biblatex}
\usepackage{nicefrac, graphicx, mathrsfs}
\usepackage{amsmath, amssymb, amsthm}
\usepackage{thmtools}
\usepackage{thm-restate}
\usepackage{physics}
\usepackage{booktabs}

\usepackage{cleveref}
\usepackage{appendix}
\usepackage{graphicx}

\usepackage{xparse}
\usepackage{xifthen}

\usepackage{tikz}
\usepackage{tikz-cd}
\usetikzlibrary{matrix,arrows}

\usepackage[all,2cell,cmtip]{xy}
\UseAllTwocells

\theoremstyle{plain}
  \newtheorem{theorem}{Theorem}[section]

  \newtheorem{proposition}[theorem]{Proposition}

\theoremstyle{definition}
  \newtheorem{definition}[theorem]{Definition}
  
  \newtheorem{ex}[theorem]{Example}
  
  \newtheorem{remark}[theorem]{Remark}

  \newcommand{\N}{\mathbb{N}}
  \newcommand{\Z}{\mathbb{Z}}
    \newcommand{\simplex}{\sigma}
\newcommand{\loss}{\ensuremath{\mathcal{L}}}

\NewDocumentCommand{\Sett}{mO{}}{\ensuremath{\qty{#1 \ifthenelse{\isempty{#2}}{}{\ | \ {#2}}}}}

\NewDocumentCommand{\Dgm}{O{}O{}O{\bullet}}{%
 \ensuremath{\ifthenelse{\isempty{#3}}{}{\mathsf{Dgm}^{\mathsf{#1}}_{#3}}\ifthenelse{\isempty{#2}}{}{\qty(#2)}}
}

\NewDocumentCommand{\ExtDgm}{O{}O{}O{\bullet}}{%
 \ensuremath{\ifthenelse{\isempty{#3}}{}{\mathsf{ExtDgm}^{\mathsf{#1}}_{#3}}\ifthenelse{\isempty{#2}}{}{\qty(#2)}}
}

\NewDocumentCommand{\ExtPH}{O{}O{}O{}}{%
 \ensuremath{\mathsf{EPH}^{\mathsf{#1}}_{#3}\ifthenelse{\isempty{#2}}{}{\qty(#2)}}
}

\newcommand{\MUTAG}{\texttt{MUTAG}\xspace}
\newcommand{\COX}{\texttt{COX2}\xspace}
\newcommand{\DHFR}{\texttt{DHFR}\xspace}
\newcommand{\NCI}{\texttt{NCI1}\xspace}
\newcommand{\IMDB}{\texttt{IMDB-B}\xspace}
\newcommand{\PRO}{\texttt{PROTEINS}\xspace}
\newcommand{\accerr}[2]{{#1}\tiny{\ensuremath{\pm{#2}}}}
\newcommand{\GIN}{{GIN}\xspace}
\newcommand{\GFL}{\texttt{GFL}\xspace}
\newcommand{\Perslay}{\texttt{Perslay}\xspace}

\newcommand{\Barc}{\textbf{Bar}}
\newcommand{\persmod}{\mathsf{V}}
\newcommand{\M}{\mathcal{M}}
\newcommand{\param}{\theta}
\newcommand{\parametrization}{\mathrm{F}}

\newcommand{\out}{{\ensuremath{\mathsf{Out}}}}
\NewDocumentCommand{\OrdPH}{O{}O{}}{%
 \ensuremath{\ifthenelse{\isempty{#2}}{\mathsf{PH}}{\mathsf{OrdPH}^{#2}}\ifthenelse{\isempty{#1}}{}{\qty(#1)}}
}

\NewDocumentCommand{\EPH}{mO{}O{\bullet}}{%
 \ensuremath{\mathsf{EPH}^{#1}_{#3}\ifthenelse{\isempty{#2}}{}{\qty(#2)}}
}
\NewDocumentCommand{\Barcode}{mO{}O{\bullet}}{%
 \ensuremath{\mathsf{Bar}^{#1}_{#3}\ifthenelse{\isempty{#2}}{}{\qty(#2)}}
}

 \NewDocumentCommand{\Bars}{O{}}{%
 \ensuremath{\mathfrak{Bar}\ifthenelse{\isempty{#1}}{}{\qty[#1]}}
}

\NewDocumentCommand{\Cone}{O{}O{}}{%
 \ensuremath{\mathsf{C}_{#2} \ifthenelse{\isempty{#1}}{}{\qty(#1)}}
}

 \NewDocumentCommand{\R}{O{}}{%
\ensuremath{\ifthenelse{\isempty{#1}}{\mathbb{R}}{\mathbb{R}^{\mathsf{#1}}}}
}

\usepackage{nicefrac, graphicx, mathrsfs}
\usepackage{amsmath, amssymb, amsthm}
\usepackage{thmtools}
\usepackage{thm-restate}

\usepackage{booktabs}
\usepackage{hyperref}
\usepackage{cleveref}

\usepackage{tikz-cd}
\usepackage{xparse}
\usepackage{xifthen}

\begin{document}
\onecolumn

\title[Optimising Graph Wavelet Persistence]{Optimisation of Spectral Wavelets for Persistence-based Graph Classification}  

\author{Ka Man Yim$^1$ \and Jacob Leygonie$^2$}
\address{Mathematical Institute, University of Oxford}
\email{$^1$yim@maths.ox.ac.uk, $^2$leygonie@maths.ox.ac.uk}

\maketitle

\begin{abstract}
 A graph’s spectral wavelet signature determines a filtration, and consequently an associated set of extended persistence diagrams. We propose a framework that optimises the choice of wavelet for a dataset of graphs, such that their associated persistence diagrams capture features of the graphs that are best suited to a given data science problem. Since the spectral wavelet signature of a graph is derived from its Laplacian, our framework encodes geometric properties of graphs in their associated persistence diagrams and can be applied to graphs without a priori node attributes. We apply our framework to graph classification problems and obtain performances competitive with other persistence-based architectures. To provide the underlying theoretical foundations, we extend the differentiability result for ordinary persistent homology to extended persistent homology.

\end{abstract}

\section{Introduction}
\subsection{Background}
Graph classification is a challenging problem in machine learning. Unlike data represented in Euclidean space, there is no easily computable notion of distance or similarity between graphs. As such, graph classification requires techniques that lie beyond mainstream machine learning techniques focused on Euclidean data. Much research has been conducted on methods such as graph neural networks (GNNs) \cite{xu2019graph} and graph kernels \cite{vishwanathan2010graph, shervashidze2009efficient} that embed graphs in Euclidean space in a consistent manner.

Recently, \emph{persistent homology}~\cite{zomorodian2005computing,edelsbrunner2008persistent} has been applied as a feature map that explicitly represents topological and geometric features of a graph as a set of \emph{persistence diagrams} (a.k.a. \emph{barcodes}). In the context of our discussion, the persistent homology of a graph $G = (V,E)$ depends on a vertex function ${f: V \to \R}$. In the case where a vertex function is not given with the data, several schemes have been proposed in the literature to assign vertex functions to graphs in a consistent way. For example, vertex functions can be constructed using local geometric descriptions of vertex neighbourhoods, such as discrete curvature  \cite{zhao2019learning}, heat kernel signatures \cite{carriere2020perslay} and Weisfeiler–Lehman graph kernels~\cite{rieck2019persistent}.

However, it is often difficult to know \emph{a priori} whether a heuristic vertex assignment scheme will perform well in addressing different data science problems. For a single graph, we can optimise the vertex function over $\qty|V|$ many degrees of freedom in $\R^V$. In recent years, there have been many other examples of persistence optimisation in data science applications. The first two examples of persistence optimisation are the computation of Fr\'echet mean of barcodes using gradients on Alexandrov spaces~\cite{turner2014frechet}, and that of point cloud inference~\cite{gameiro2016continuation}, where a point cloud is optimised so that its barcode fits a target fixed barcode. The latter is an instance of topological inverse problems (see~\cite{oudot2020inverse} for a recent overview of such). Another inverse problem is that of surface reconstruction~\cite{bruel2020topology}. Besides, in the context of shape matching~\cite{poulenard2018topological}, persistence optimisation is used in order to learn an adequate function between shapes. Finally, there are also many recent applications of persistence optimisation in Machine Learning, such as the incorporation of topological information in Generative Modelling~\cite{moor2020topological, hofer2019connectivity,gabrielsson2020topology} or in Image Segmentation~\cite{hu2019topology,clough2019explicit}, the design of topological losses for Regularization in supervised learning~\cite{chen2019topological} or for dimension reduction~\cite{kachan2020persistent}. 

Each of these applications can be thought of as minimising a certain \emph{loss} function over a manifold~$\M$ of parameters:
\begin{equation*} 
    \begin{tikzcd}
    \min_{\theta\in \M }\loss(\theta),
    \end{tikzcd} 
\end{equation*}
where $\loss: \M\rightarrow \Barc^N \rightarrow \R$ factors through the space~$\Barc^N$ of~$N$-tuples of barcodes. The aim is to find the parameter~$\theta$ that best fits the application at hand. 
Gradient descent is a very popular approach in minimisation, but it requires the ability to differentiate the loss function. 
In fact, \cite{leygonie2019framework} provide notions of differentiability for maps in and out~$\Barc$ that are compatible with smooth calculus, and show that the loss functions~$\loss$ corresponding the applications cited in the above paragraph are generically differentiable. The use of (stochastic) gradient descent is further legitimated by \cite{carriere2020note}, where convergence guarantees on persistence optimisation problems are devised, using a recent study of stratified non-smooth optimisation problems~\cite{davis2020stochastic}. In practice, the minimisation of~$\loss$ can be unstable due to its non-convexity and partial non-differentiability. Some research has been conducted in order to smooth and regularise the optimisation procedure~\cite{solomon2020fast,corcoran2020regularization}.

In a supervised learning setting, we want to optimise our vertex function assignment scheme over many individual graphs in a dataset. Since graphs may not share the same vertex set and come in different sizes, optimising over the $\qty|V|$ degrees of freedom of any one graph is not conducive to learning a vertex function assignment scheme that can generalise to another graph. The degrees of freedom in any practical vertex assignment scheme should be independent of the number of vertices of a graph. However, a framework for  parametrising and optimising the vertex functions of many graphs over a common parameter space $\M$ is not immediately apparent.

The first instance of a graph persistence optimisation framework (\GFL) \cite{hofer2020graph} uses a one layer graph isomorphism network (\GIN) \cite{xu2019graph} to parametrise vertex functions. The \GIN learns a vertex function by exploiting the local topology around each vertex. In this paper, we propose a different framework for assigning and parametrising vertex functions, based on a graph's Laplacian operator. Using the Laplacian, we can explicitly take both local and global structures of the graph into consideration in an interpretable and transparent manner.  \\

\subsection{Outline and Contributions} We address the issue of vertex function parametrisation and optimisation using \emph{wavelet signatures}. Wavelet signatures are vertex functions derived from the eigenvalues and eigenvectors of the graph Laplacian and encode multiscale geometric information about the graph \cite{li2013multiresolution}. The wavelet signature of a graph is dependent on a choice of wavelet $g: \R \to \R$, a function on the eigenvalues of the graph's Laplacian matrix. We can thus obtain a parametrisation of vertex functions for any graph $\parametrization : \M \to \R^V$ by parametrising $g$. Consequently, the \emph{extended persistence} of a graph \--- which has only four non-trivial persistence diagrams \--- can be varied over the parameter space $\M$. If we have a function $\out: \Barc^4 \to \R$ on persistence diagrams that we wish to minimise, we can optimise over $\M$ to minimise the loss function
\begin{equation} \label{eq:pipeline}
    \begin{tikzcd}
    \loss : \M  \arrow[r, "\parametrization"] & \R^{V} \arrow[r, "\ExtPH"] &\Barc^4 \arrow[r, "{\out}"] & \R.
    \end{tikzcd} 
\end{equation}
If $\loss$ is generically differentiable, we can optimise the wavelet signature parameters $\theta \in \M$ using gradient descent methods. We illustrate an application of this framework to a graph classification problem in \Cref{fig:pipeline}, where the loss function $\loss$ is the classification error of a graph classification prediction model based on the graph's extended persistence diagrams. 

\begin{figure}
    \centering
    \includegraphics[width = \textwidth]{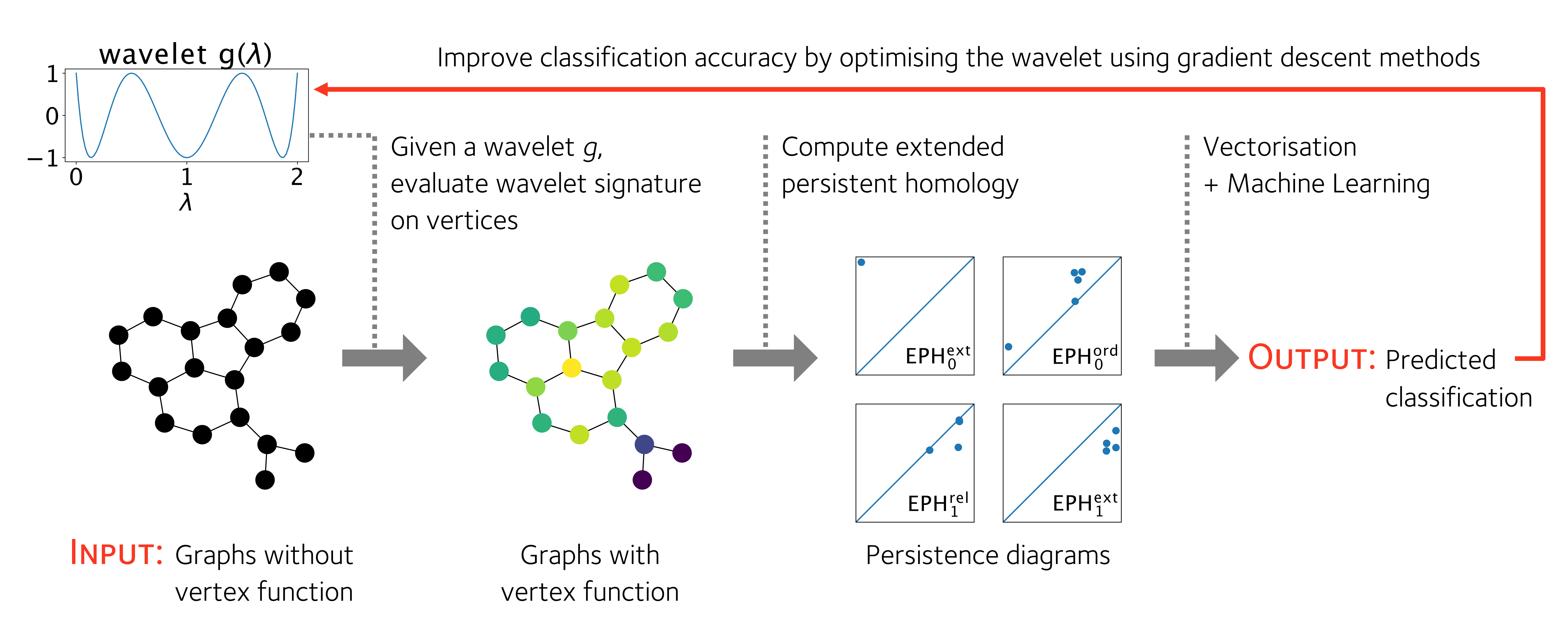}
    \caption{Given a wavelet $g: \R \to \R$, we can equip any graph with a non-trivial vertex function. This allows us to compute the extended persistence diagrams of a graph and use the diagrams as features of the graph to predict a graph's classification in some real world setting. The wavelet $g$ can be optimised to improve the classification accuracy of a graph classification pipeline based on the extended persistence diagrams of a graph's vertex function. }
    \label{fig:pipeline}
\end{figure}

In \Cref{sec:wavelet}, we describe the assignment of vertex functions $\parametrization: \M \to \R^V$ by reviewing the definition of wavelet signatures. While spectral wavelets have been used in graph neural network architectures that predict vertex features \cite{xu2019graph} and compress vertex functions \cite{rustamov2019wavelets}, they have not been considered in a persistent homology framework for graph classification. We describe several ways to parametrise wavelets. We also show in \Cref{prop:wavelet} that wavelet signatures are independent of the choice of eigenbasis of the graph Laplacian from which it is derived, ensuring that it is well-defined. We prove this result in \Cref{app:wav}.

In \Cref{sec:eph}, we describe the theoretical basis for optimising the \emph{extended} persistent homology of a vertex function $\ExtPH: \R^V \to \Barc^4$ and elucidate what it means for $\loss$ to be differentiable. In  \Cref{proposition_extended_generically_differentiable}, we generalise the differentiability formalism of ordinary persistence \cite{leygonie2019framework} to extended persistence. We prove this result in \Cref{app:diff}

Finally, in \Cref{sec:expts}, we apply our framework to graph classification problems on several benchmark datasets. We show that our model is competitive with state-of-the-art persistence-based models. In particular, optimising the vertex function appreciably improves the prediction accuracy on some datasets.

\section{Filter Function Parametrization} \label{sec:wavelet}
We describe our recipe for assigning vertex functions to any simplicial graph $G = (V,E)$ based on a parametrised spectral wavelet, the first part $\parametrization$ of the loss function 
\begin{equation} 
    \begin{tikzcd}
    \loss : \M  \arrow[r, "\parametrization"] & \R^{V} \arrow[r, "\ExtPH"] &\Barc^4 \arrow[r, "{\out}"] & \R
    \end{tikzcd}. \tag{\Cref{eq:pipeline} recalled}
\end{equation}
Our recipe is based on a graph's wavelet signature, a vertex function derived from the graph's Laplacian. The wavelet signature also depends on a so-called `wavelet function' in $g : \R \to \R$, which is independent of the graph. By modulating the wavelet, we can jointly vary the wavelet signature across many graphs. We parametrise the wavelet using a finite linear combination of basis functions, such that the wavelet signature can be manipulated in a computationally tractable way. In the following section, we define the wavelet signature and describe our linear approach to wavelet parametrisation. \\

\subsection{Wavelet Signatures} The wavelet signature is a vertex function initially derived from wavelet transforms of vertex functions on graphs \cite{hammond2011wavelets}, a generalisation of wavelet transforms for square integrable functions on Euclidean space \cite{graps1995introduction, chui2016introduction} for signal analysis \cite{akansu2001multiresolution}. Wavelet signatures for graphs have been applied to encode geometric information about meshes of 3D shapes \cite{akansu2001multiresolution, li2013multiresolution}. Special cases of wavelets signatures, such as the heat kernel signature \cite{sun2009concise} and wave kernel signature \cite{aubry2011wave}, have also been applied to describe graphs and 3D shapes \cite{bronstein2010scale, hu2014stable}.

The wavelet signature of a graph is constructed from the graph's Laplacian operator. A graph's normalised Laplacian $L \in \R^{V \times V}$ is a symmetric positive semi-definite matrix, whose entries are given by
\begin{equation}
    L_{uv} = 
    \begin{cases} 
    1 & u=v \\
    -\frac{1}{\sqrt{k_u k_v}} & (u,v) \in E \\
    0  & \text{otherwise}
    \end{cases}
\end{equation}
where $k_u$ is the degree of vertex $u$. The Laplacian's eigenvalues $\lambda$ and eigenvectors $ \vb*{\phi}$ are known to encode various topological and geometric information about the graph \cite{chung1997spectral, biyikoglu2007laplacian}; for example, the number of zero eigenvalues corresponds to the number of connected components of the graph. The spectrum of the normalised Laplacian have real eigenvalues in  $\qty[0,2]$ \cite{chung1997spectral}. As such, any function $g: \R \to \R$ evaluated on the eigenvalues need only be defined on $\qty[0,2]$. Moreover, functions on a compact domain are easily parametrised using convenient bases. 

\begin{restatable}{definition}{defwavelet}[Wavelet Signature \cite{li2013multiresolution}] \label{def:wavelet} Let $L \in \R^{V\times V} $ be the normalised Laplacian of a simplical graph $G = (V,E)$. Let $\vb*{\phi}_1,\ \ldots,\ \vb*{\phi}_{\qty|V|}$ be an orthonormal eigenbasis for $L$ and $\lambda_1, \ldots ,  \lambda_{\qty|V|}$ be their corresponding eigenvalues. The wavelet signature $W: \R^\qty[0,2] \to \R^V$ maps a function $g: [0,2] \to \R$, which we refer to as a \emph{wavelet}, to a vertex function $W(g) \in \R^V$ linearly, where the value of $W(g)$ on vertex $v$ is given by
\begin{equation} \label{eq:WSdef}
    W(g)_v = \sum_{i=1}^{\qty|V|} g(\lambda_i)\qty(\vb*{\phi}_{i})_v^2,
\end{equation}
and $\qty(\vb*{\phi}_{i})_v$ denotes the component of eigenvector $\vb*{\phi}_{i}$ corresponding to vertex $v$.
\end{restatable}

If the eigenvalues of $L$ have geometric multiplicity one (i.e. their eigenspaces are one dimensional), then the orthonormal eigenvectors are uniquely defined up to a choice of sign. It is then apparent from \Cref{eq:WSdef} that the wavelet signature is independent of the choice of sign. However, if some eigenvalues have geometric multiplicity greater than one, then the orthonormal eigenvectors of $L$ are uniquely defined up to orthonormal transformations in the individual eigenspaces. However, the wavelet signature is well-defined even when the multiplicities of eigenvalues are greater than one. This is the content of the next Proposition, whose proof is deferred to \Cref{app:wav}.

\begin{restatable}{proposition}{waveletprop}
\label{prop:wavelet}
The wavelet signature of a graph is independent of the choice of orthonormal eigenbasis for the Laplacian.
\end{restatable}
\begin{remark}  \label{remark:wavelet_interp}
In addition to the traditional view of wavelets from a spectral signal processing  perspective \cite{hammond2011wavelets}, we can also relate the wavelet signature of a vertex $v$ to the degrees of vertices in some neighbourhood of $v$ prescribed by $g$. Consider a wavelet~$g: [0,2] \to \R$. On a finite graph $G$, the normalised Laplacian $L$ has at most $\qty|V|$ many distinct eigenvalues. As such, there exists a polynomial $\hat{g}(x) = \sum_{n=0}^p a_n x^n$ of finite order that interpolates $g$ at the eigenvalues $g(\lambda_i) = \hat{g}(\lambda_i)$. Therefore, $W\qty(g) = W\qty(\hat{g})$. Moreover, the vertex values assigned by $W\qty(\hat{g})$ are the diagonal entries of the matrix polynomial $\hat{g}(L)$:
\begin{equation}\label{eq:wavelet_polyseries}
   \hat{g}(L)_{vv} = \sum_{n=0}^p  a_n \qty(L^n)_{vv} = \sum_{i=1}^{\qty|V|} \hat{g}(\lambda_i)\qty(\vb*{\phi}_{i})_v^2 = \sum_{i=1}^{\qty|V|} g(\lambda_i)\qty(\vb*{\phi}_{i})_v^2 = W(g)_{vv}.
\end{equation}
Furthermore, we can also write the matrix polynomial $\hat{g}(L)$ as a matrix polynomial in $A = I-L$, the \emph{normalised adjacency matrix}. From the definition of $L$, we can compute the diagonal entry of a monomial $A^r$  corresponding to vertex $v$ as an inverse degree weighted count of paths\footnote{Here a path refers to a sequences of vertices that are connected to the next vertex in the sequence by an edge.} $\qty[v_0, v_1, \ldots, v_r]$ on the graph which begin and end on vertex $v = v_0 = v_{r}$~\cite{newman2018networks}:
\begin{equation}
    \qty(A^r)_{vv} =  \frac{1}{k_v} \sum_{\qty[v, v_1, \ldots, v_{r-1}, v]}  \qty(\prod_{l = 1}^{r-1} \frac{1}{k_{v_l}} ).
\end{equation}
By expressing the wavelet signature as a matrix polynomial in $A$, we see that $g$ controls how information at different length scales of the graph contribute to the wavelet signature. For instance, if $g$ were an order~$p$ polynomial, then $W(g)_v$ only takes the degrees of vertices that are $\lfloor p/2 \rfloor$ away from $v$ into account. As a corollary, since $W(g)$ can be specified by replacing $g$ with a polynomial $\hat{g}$ of order at most $\qty|V|-1$, the wavelet signature at a vertex is only dependent on the subgraph of $G$ that is within $\lfloor\qty|V|-1\rfloor/2$ steps away from $v$. 
\end{remark}

\subsection{Parametrising the Wavelet}  \label{sec:param}
We see from \Cref{remark:wavelet_interp} that the choice of wavelet $g$ determines how the topology and geometry of the graph is reflected in the vertex function. Though the space of wavelets is potentially infinite dimensional, here we only consider wavelets $g_\theta(x)$ that are parametrised by parameters $\theta$ in a finite dimensional manifold, so that we can easily optimise them using computational methods. In particular, we focus on wavelets written as a linear combination of~$m$ \emph{basis functions} $h_1, \ldots , h_m : [0,2] \to \R$
\begin{equation}\label{eq:wavelet_coeff}
   g_\theta(x) :=  \sum_{j=1}^m \theta_j h_j(x)
\end{equation}
This parametrisation of wavelets in turn defines a parametrisation of vertex functions $\parametrization: \R^m \to \R^V$ for our optimisation pipeline in \cref{eq:pipeline}
\begin{equation}
   \parametrization:\quad  \theta \in \R^m \quad  \longmapsto \quad \parametrization(\theta) := W\qty(g_\theta) \quad \in \R^V.
\end{equation}
Since $W(g)$ is a linear function of the wavelet $g$, $\parametrization$ is a linear transformation:
\begin{equation}
    \parametrization(\theta)  = W\qty(\sum_{j=1}^m \theta_j h_j(x)) = \sum_{j=1}^m \theta_j W\qty(h_j).
\end{equation}
We can write $\parametrization$ as a $\qty|V|\times m$ matrix acting on a vector $\qty[\theta_1, \ldots \theta_m]^\intercal \in \R^m$, whose columns are the vertex functions $W\qty(h_j)$.

\begin{ex}[Chebyshev Polynomials]
Any Lipschitz continuous function on an interval can be well approximated by truncating its Chebyshev series at some finite order~\cite{trefethen1997numerical}. 
The Chebyshev polynomials  $T_n: [-1, 1] \to \R$
\begin{equation}\label{eq:chebs}
    T_n\qty(x) = \cos(n\arccos(x)) \quad n \in \N_{\geq 0} .
\end{equation}
form an orthonormal set of functions. We can thus consider $h_j(\lambda) = T_j(\lambda - 1),\ j = 0, 2, \ldots ,m$ as a naïve basis for wavelets. We exclude $T_1(x) = x$ in the linear combination as $W(T_1(1-x)) = 0$ for graphs without self loops.
\end{ex}

\begin{ex}[Radial Basis Functions]
In the machine learning community, a \emph{radial function} refers loosely to a continuous monotonically decreasing function $\rho: \R_{\geq 0} \to \R_{\geq 0} $. There are many possible choices for $\rho$, for example, the inverse multiquadric
\begin{equation} \label{eq:rbf}
    \rho(r) = \qty(\qty(\frac{r}{\epsilon})^2 + 1)^{-\frac{1}{2}}
\end{equation}
where $\epsilon \neq 0$ is a width parameter. We can obtain a naïve wavelet basis $h_j(x) = \rho\qty(\norm{x - x_j})$ using copies of $\rho$ offset by a collection of centroids $x_j \in \R$ along $\R$. In general, the centroids are parameters that could be optimised, but we fix them in this study. This parametrisation can be considered as a \emph{radial basis function neural network}. RBNNs are well-studied in function approximation and subsequently machine learning; we refer readers to \cite{chen1991orthogonal,park1991universal} for further details.
\end{ex}

\subsection{The Choice of Wavelet Basis}
The choice of basis functions determines the space of wavelet signatures and also the numerical stability of the basis function coefficients which serve as the wavelet signature parameters. The stability of the parametrisation depends on the graphs as much as the choice of wavelet basis $h_1, \ldots , h_m$. We can analyse the stability of a parametrisation $\parametrization$ by its the singular value decomposition
\begin{equation}\label{eq:svd}
   \parametrization = \sum_{k=1}^r \sigma_k \vb{u}_k \vb{v}_k^\intercal
\end{equation} 
where $\sigma_1, \ldots, \sigma_r$ are the non-zero singular values of the matrix, and $\vb{u}_k \in \R^\qty|V|$ and $\vb{v}_k \in \R^m$ are orthonormal sets of vectors respectively. If the distribution of singular values span many orders of magnitude, we say the parametrisation is \emph{ill-conditioned}. An ill-conditioned parametrisation interferes with the convergence of gradient descent algorithms on a loss function evaluated on wavelet signatures. We discuss the relationship between the conditioning of $\parametrization$ and the stability of gradient descent in detail in \Cref{rem:conditioning}. 

We empirically observe that the coefficients of a naïve choice of basis functions, such as Chebyshev polynomials or radial basis functions, are numerically ill-conditioned. In \Cref{fig:singular_values}, we can see that the singular values of radial basis function and Chebyshev polynomial parametrisations respectively are distributed across a large range on the logarithmic scale for some datasets of graphs in machine learning. We address this problem by picking out a new wavelet basis
\begin{equation} \label{eq:stable_basis}
    h_k'(x) = \frac{1}{\sigma_k}\sum_{j=1}^m \qty(\vb{v}_k)_j h_j(x) \qc k = 1,\ldots , r,
\end{equation}
where $\sigma_k$ are the singular values of $\parametrization$ and $\vb{v}_k$ are the associated vectors in $\R^m$ from the singular value decomposition of matrix $\parametrization$ in \cref{eq:svd}. Then the parametrisation $\parametrization' : \R^r \to \R^V$ 
\begin{equation} \label{eq:new_param}
      \parametrization' (\theta') = \sum_{k=1}^r \theta'_k W(h'_k).
\end{equation}
have singular values equal to one, since this is a linear combination of orthonormal vectors $\vb{u}_k \in \R^V$:
\begin{equation}
    W(h'_k) =  \sum_{j = 1}^m \frac{1}{\sigma_k} \qty(\vb{v}_k)_j W(h_j) = \frac{1}{\sigma_k}\parametrization \vb{v}_k = \vb{u}_k.
\end{equation}
As an example, we plot the new wavelet basis $h'_k$ derived from a twelve parameter radial basis function parametrisation for the \MUTAG dataset in \Cref{fig:mutag_basis} in \Cref{app:wav}.

\begin{remark}[Learning a Wavelet Basis for Wavelet Signatures on Multiple Graphs]
In the case where the wavelet coefficients parametrise the wavelet signatures over graphs $G_1,\ldots, G_N$, we can view the maps $F_1, \ldots, F_N$ that map wavelet basis coefficients to vertex functions of graphs $G_1, \ldots, G_N$ respectively as a parametrisation for the disjoint union $\bigsqcup_i G_i$:
\begin{equation} \label{eq:param_whole_dataset}
  f =  \bmqty{f_1 \\ \vdots  \\ f_N}  = \bmqty{\parametrization_1 \\ \vdots  \\ \parametrization_N}\theta =: \parametrization \theta .
\end{equation}
We can then perform a singular value decomposition of the parametrisation $\parametrization$ on $\bigsqcup_i G_i$ and derive a new, well-conditioned basis.
\end{remark}

\begin{remark}[Why the Conditioning of $\parametrization$ Matters]\label{rem:conditioning} Let us optimise a loss function $\loss$ on the parameter space of wavelet coefficients $\theta$ using a gradient descent algorithm. In a gradient descent step of step size~$s$, the wavelet coefficients are updated to~$\theta \mapsto \theta - s \grad_\theta \loss$. Using the singular value decomposition of $\parametrization$ (\cref{eq:svd}), we can write 
\begin{equation}\label{eq:theta_svd}
     \grad_\theta \loss = {\grad_\theta f}^\intercal    {\grad_f \loss} = \parametrization^\intercal {\grad_f \loss} = \sum_{k=1}^r \sigma_k \expval{\grad_f \loss, \vb{u}_k} \vb{v}_k.
\end{equation}
The change in the vertex function is simply the matrix $\parametrization$ applied to the change in wavelet parameters. Hence the vertex function is updated to~$ f \mapsto f - s \parametrization \grad_\theta \loss$, where 
\begin{equation} \label{eq:vertfunc_svd}
    \parametrization \grad_\theta \loss  = \sum_{k=1}^r \sigma_k \expval{\grad_f \loss, \vb{u}_k} \parametrization \vb{v}_k=  \sum_{k=1}^r \sigma_k^2 \expval{\grad_f \loss, \vb{u}_k} \vb{u}_k .
\end{equation}
If the loss function $\loss$ has large second derivatives\--- for example, due to nonlinearities in the function on persistence diagrams $\out: \Barc^4 \to \R$ \--- the projections $\expval{\grad_f \loss,\ \vb{u}_k}$ in \cref{eq:theta_svd,eq:vertfunc_svd} may change dramatically from one gradient descent update to another. If the smallest singular value is much smaller than the largest, then updates to the wavelet signature can be especially unstable throughout the optimisation process. This source of instability can be removed if we choose a parametrisation with uniform singular values $\sigma_k = 1$. In this case,  the update to $f$ is simply the projection of $\grad_f \loss$ onto the space of wavelet signatures spanned by $\vb{u}_1, \ldots ,  \vb{u}_r$, without any distortion introduced by non-uniform singular values:
\begin{equation} 
    f \mapsto f -  s \sum_{k=1}^r  \expval{\vb{u}_k,  \grad_f \loss} \vb{u}_k.  
\end{equation}
\end{remark}

\section{Extended Persistent Homology} 

\label{sec:eph} The homology of a given graph is a computable vector space whose dimension counts the number of connected components or cycles in the graph. Finer information can be retained by filtering the graph and analysing the evolution of the homology throughout the filtration. This evolution is described by a set of \emph{extended persistence diagrams} (a.k.a. \emph{extended barcodes}), a multiset of points $\expval{b,d}$ that record the birth $b$ and death of homological features in the filtration. In this section, we begin by summarising these constructions. We refer the reader to~\cite{zomorodian2005computing}, \cite{edelsbrunner2008persistent}, and \cite{cohen2007stability} for full treatments of the theory of Persistence. 

Compared to \emph{ordinary persistence}, extended persistence is a more informative and convenient feature map for graphs. Extended persistence encodes strictly more information than ordinary persistence. For instance, the cycles of a graph are represented as points with $d = \infty$ in ordinary persistence. Thus, only the birth coordinate $b$ of such points contain useful information about the cycles. In contrast, the corresponding points in extended persistence are each endowed with a finite death time $d$, thus associating extra information to the cycles. The points at infinity in ordinary persistence also introduce obstacles to vectorisation procedures, as often arbitrary finite cutoffs are needed to `tame' the persistence diagrams before vectorisation. \\

\subsection{Extended Persistent Homology} Let~$G=(V,E)$ be a finite graph without double edges and self-loops. For the purposes of this paper, the associated \emph{extended persistent homology} is a map \[\ExtPH[][][]: \R^V \to \Barc^{4}\] from functions $f \in \R^V$ on its vertices to the space of four \emph{persistence diagrams} or \emph{barcodes}, which we define below. The map arises from a \emph{filtration} of the graph, a sequential attachment of vertices and edges in ascending or descending order of~$f$. We extend~$f$ on each edge~$e=(v,v')$ by the maximal value of~$f$ over the vertices~$v$ and~$v'$, and we then let~$G_t \subset G$ be the sub graph induced by vertices taking value less than~$t$. Then we have the following sequence of inclusions:

\begin{equation}
\label{eq_ordinary_filtration}
\begin{tikzcd}[column sep= normal]
\emptyset \arrow[r] & \cdots \arrow[r] & G_s \arrow[r, "s \leq t"] & G_t \arrow[r] & \cdots \arrow[r]  & G.
\end{tikzcd}
\end{equation}

Similarly, the sub graphs $G^t \subset G$ induced by vertices taking value greater than~$t$ assemble into a sequence of inclusions:

\begin{equation}
\label{eq_relative_filtration}
\begin{tikzcd}[column sep=normal]
G & \arrow[l]  \cdots & \arrow[l] G^s  & \arrow[l, "s \leq t"'] G^t  & \arrow[l] \cdots  & \arrow[l]  \emptyset
\end{tikzcd}.
\end{equation}

The changes in the topology of the graph along the filtration in ascending and descending order of~$f$ can be detected by its \emph{extended persistence module}, indexed over the poset $\R \cup \qty{\infty} \cup \R[op]$:
\begin{equation}
\label{eq_extended_persistent_homology}
\persmod_{p}(f) :\
\begin{tikzcd}[column sep=small]
0 = H_p(\emptyset) \arrow[r] & \cdots \arrow[r] & H_p(G_s) \arrow[r, "s \leq t"] & H_p(G_t) \arrow[r] & \cdots \arrow[r]  & H_p(G) \arrow[d, "\cong"] \\
0 = H_p(G,G) & \arrow[l]  \cdots & \arrow[l] H_p(G, G^s)  & \arrow[l, "s \leq t"'] H_p(G, G^t)  & \arrow[l] \cdots  & \arrow[l]  H_p(G, \emptyset)
\end{tikzcd},
\end{equation}
where $H_p$ is the singular (relative) homology functor in degre~$p\in {0,1}$ with coefficients in a fixed field, chosen to be $\Z/2\Z$ in pratice. In general terms, the modules~$\persmod_{0}(f)$ and~$\persmod_{1}(f)$ together capture the evolution of the connected components and loops in the sub graphs of~$G$ induced by the function~$f$. 

Each module $\persmod_{p}(f)$ is completely characterised by a finite multi-set $\ExtPH[][f][p]$ of pairs of real numbers $\expval{b,d}$ called \emph{intervals} representing the birth and death of homological features. Following~\cite{cohen2009extending}, the intervals in~$\ExtPH[][f][p]$ are further partitioned according to the type of homological feature they represent:

\begin{equation}
    \ExtPH[][f][p]= \underbrace{\Sett{\expval{b,d}}[b<d<\infty ]}_{ = \ExtPH[ord][f][p]} \sqcup \underbrace{\Sett{\expval{b,d}}[b <\infty< d ]}_{ = \ExtPH[ext][f][p]} \sqcup \underbrace{\Sett{\expval{b,d}}[\infty<b<d ]}_{ = \ExtPH[rel][f][p]} .
\end{equation}

Each of the three finite multiset $\ExtPH[k][f][p]$, for $k\in \{\mathrm{ord},\mathrm{ext},\mathrm{rel}\}$, is an element in the space~$\Barc$ of so-called \emph{barcodes} or \emph{persistence diagrams}. However,~$\ExtPH[\mathrm{rel}][f][0]$ and~$\ExtPH[\mathrm{ord}][f][1]$ being trivial for graphs, we refer to the collection of four remaining persistence diagrams

\begin{equation}
    \ExtPH[][f][] = \qty[\ExtPH[\mathrm{ord}][f][0],\ExtPH[\mathrm{ext}][f][0],\ExtPH[\mathrm{ext}][f][1],\ExtPH[\mathrm{rel}][f][1] ] \in \Barc^{4}
\end{equation}

as the \emph{extended barcode} or \emph{extended persistence diagram} of~$f$. We have thus defined the \emph{extended persistence map} 
\[ \ExtPH[][][]: \R^V \to \Barc^{4}.\]

\begin{remark}
\label{remark_ordinary_persistence}
If we only apply homology to the filtration of~Eq.~\eqref{eq_ordinary_filtration}, we get an \emph{ordinary persistence module} indexed over the real line, which is essentially the first row in Eq.~\eqref{eq_extended_persistent_homology}. This module is characterised by a unique barcode~$\OrdPH_p(f)\in \Barc$. We refer to the map

\begin{equation}
\label{equation_ordinary_persistence}
\OrdPH: f\in \R^V \longmapsto \qty[\OrdPH_0(f),\OrdPH_1(f)] \in \Barc^{2}
\end{equation}
as the \emph{ordinary persistence map}.

\end{remark}

\subsection{Differentiability of Extended Persistence}

The extended persistence map can be shown to be locally Lipschitz by the Stability theorem~\cite{cohen2009extending}. The Rademacher theorem states that any real-valued function that is locally Lipschitz is differentiable on a full measure set. Thus, so is our loss function
\begin{equation} 
    \begin{tikzcd}
    \loss : \M  \arrow[r, "\parametrization"] & \R^{V} \arrow[r, "\ExtPH"] &\Barc^4 \arrow[r, "{\out}"] & \R
    \end{tikzcd}. \tag{\Cref{eq:pipeline} recalled}
\end{equation}
 as long as~$\out$ and~$\parametrization$ are smooth or locally Lipschitz\footnote{In practice, a locally Lipschitz $\out$ can be constructed out of Lipschitz stable vectorisation methods, such as Persistence Landscapes \cite{bubenik2015statistical} and Persistence Images \cite{adams2017persistence}.}. If a loss function~$\loss$ is locally Lipschitz, we can use stochastic gradient descent as a paradigm for optimisation. Nonetheless, the theorem above does not rule out dense sets of non differentiability in general.

In this section, we show that the set where~$\ExtPH[][][]$ is not differentiable is not pathological. Namely, we show that~$\ExtPH[][][]$ is \emph{generically} differentiable, i.e. differentiable on an open dense subset. This property guarantees that local gradients yield reliable descent directions in a neighbourhood of the current iterate. We recall from~\cite{leygonie2019framework} the definition of differentiability for maps to barcodes.

We call a map $\parametrization: \M \rightarrow \R^V$ a \emph{parametrisation}, as it corresponds to a selection of filter functions over~$G$ parametrised by the manifold~$\M$. Then~$B := \ExtPH \circ \parametrization$ is the barcode valued map whose differentiability properties are of interest in applications.

\begin{definition}
\label{definition_differentiability}
A map $B: \M \rightarrow \Barc$ on a smooth manifold $\M$ is said to be differentiable at $\param\in \M$ if for some neighbourhood $U$ of $\param$, there exists a finite collection of differentiable maps~\footnote{By convention, a differentiable map that takes the value $\infty$ is constant.} $b_i,d_i: U \rightarrow \R \cup \{\infty\}$, called a \emph{local coordinate system} for~$B$ at~$\param$, such that
\[\forall \param'\in U, \, \, B(\param')=  \qty{\expval{b_i(\theta'), d_i(\theta')}\,  | \, b_i(\theta') \neq d_i(\theta')}.\]

For $N\in \N$, we say that a map $B: \M \rightarrow \Barc^N$ is differentiable at $\param$ if all its components are so.

\end{definition}

In~\cite{leygonie2019framework}, it is proven that the composition $\OrdPH \circ \parametrization$ is generically differentiable as long as $\parametrization$ is so. It is possible to show that $\ExtPH[][][] \circ \parametrization$ is generically differentiable along the same lines, but we rather provide an alternative argument in the appendix. Namely, we rely on the fact that the extended persistence of~$G$ can be decoded from the ordinary persistence of the cone complex~$\Cone[G][]$, a connection first noted in~\cite{cohen2009extending} for computational purposes. 

\begin{proposition}
\label{proposition_extended_generically_differentiable}

Let $\parametrization: \M \rightarrow \R^V$ be a generically differentiable parametrisation. Then the composition $\ExtPH[][][] \circ \parametrization$ is generically differentiable.

\end{proposition}

For completeness, the proof provided in the appendix treats the general case of a finite simplicial complex~$K$ of arbitrary dimension. 

\section{Binary graph classification} \label{sec:expts}
We investigate whether optimising the extended persistence of wavelet signatures can be usefully applied to graph classification problems, where persistence diagrams are used as features to predict discrete, real life attributes of networks.  In this setting, we aim to learn $\theta \in \M$ that minimise the classification error of graphs over a training dataset. 

We apply our wavelet optimisation framework to classification problems on the graph datasets \MUTAG \cite{1debnath1991structure,23}, \COX \cite{7sutherland2003spline}, \DHFR \cite{7sutherland2003spline}, \NCI \cite{8wale2008comparison, 22shervashidze2011weisfeiler}, \PRO \cite{4borgwardt2005protein, 6dobson2003distinguishing} and \IMDB \cite{14yanardag2015deep}. The former five datasets are biochemical molecules and \IMDB is a collection of social ego networks. In our models, we use persistence images \cite{adams2017persistence} as a fixed vectorisation method and use a feed forward neural network to map the persistence images to classification labels. We also include the eigenvalues of the graph Laplacian as additional features; model particulars are described in the sections below.

To illustrate the effect of wavelet optimisation on different classification problems,  we also perform a set of \emph{control experiments} where for the same model architecture, we fix the wavelet and only optimise the parameters of the neural network. The control experiment functions as a baseline against which we assess the efficacy of wavelet optimisation.
 
We benchmark our results with two existing persistence based architectures, \texttt{PersLay} \cite{carriere2020perslay} and \texttt{GFL} \cite{hofer2020graph}. \Perslay optimises the vectorisation parameters and use two heat kernel signatures as fixed rather than optimisable vertex functions for computing extended persistence. \GFL optimises and parametrises vertex functions using a graph isomorphism network \cite{xu2019graph}, and computes \emph{ordinary} sublevel and superlevel set persistence instead of extended persistence.  \\

\subsection{Model Architecture}
We give a high level description of our model and relegate details and hyperparameter choices of the vectorisation method and neural network architecture to \cref{app:expt}. In our setting, the extended persistence diagrams of the optimisable wavelet signatures for each graph are vectorised as persistence images. We also include the static persistence images of a \emph{fixed} heat kernel signature, $W(e^{-0.1x})$, as an additional set of features, alongside some non-persistence features. Both the optimised and static persistence diagrams are transformed into the persistence images using identical hyperparameters. We feed the optimisable and static persistence images into two separate convolutional neural networks (CNNs) with the same architecture. Similarly, we feed the non-persistence features as a vector into a separate multilayer perceptron. The outputs of the CNNs are concatenated with the outputs of the multi-layer perceptron. Finally, an affine transformation sends the concatenated vector to a real number whose sign determines the binary classification.
\\

\subsubsection{Wavelet Parametrisation}
We choose a space of wavelets spanned by 12 inverse multiquadric radial basis functions 
\begin{equation} \label{eq:rbf_actual}
    h_j(x) = \qty(\qty(\frac{x - x_j}{\epsilon})^2 + 1)^{-\frac{1}{2}}
\end{equation}
whose centroids $x_j$ are located at~$x_j = 2(j-1)/9,\ j = 0, \ldots ,11$. The width parameter is chosen to be the distance between the centroids, $\epsilon = 2/9$. On each dataset, we derive a numerically stable parametrisation using the procedure described in \Cref{sec:param}; the parameters we optimise are the coefficients of the new basis given by \cref{eq:stable_basis}. We initialise the parameters by fitting them via least squares to the heat kernel signature $W(e^{-10x})$ on the whole dataset of graphs. 
\\

\subsubsection{Non-Persistence Features}
\label{sec:non-pers}
We also incorporate the eigenvalues of the normalised Laplacian as additional, fixed features of the graph. Since the number of eigenvalues for a given graph is equal to the number of vertices, it differs between graphs in the same dataset. To encode the information represented in the eigenvalues as a fixed length vector, we first sort the eigenvalues into a time-series; we then compute the log path signature of the time series up to level four, which is a fixed length vector in $\R^8$. The log-signature captures the geometric features of the path; we refer readers to \cite{chevyrev2016primer} for details about path signatures.  For \IMDB in particular, we also include the maxima and minima of the heat kernel signatures $W(e^{-10x})$ and $W(e^{-0.1x})$ respectively of each graph. \\

\subsection{Experimental set up}
We employ a \emph{10 ten-fold} test-train split scheme on each dataset to measure the accuracy of our model. Each ten-fold is a set of ten experiments, corresponding to a random partition of the dataset into ten portions. In each experiment, a different portion is selected as the test set while the model is trained on the remaining nine portions. We perform 10 ten-folds to obtain a total of 10 $\times$ 10 experiments, and report the accuracy of the classifier as the average accuracy over 100 such experiments. The epochs at which the accuracies were measured are specified in \Cref{tab:our_full_results_epochs}.

Across all experiments, we use binary cross entropy as the loss function. We use the \texttt{Adam} optimiser \cite{kingma2014adam} with learning rate $\texttt{lr = 1e-3}$ to optimise the parameters of the neural network. The wavelet parameters are updated using stochastic gradient descent with learning rate $\texttt{lr = 1e-2}$, for all datasets except for \IMDB, where the learning rate is set to $\texttt{lr} = 1e-1$. The batch sizes for each experiment are shown in \Cref{tab:batchsize}. In all experiments, we stop the optimisation of wavelet parameters at epoch 50 while the neural network parameters continue to be optimised.

We use the \href{https://gudhi.inria.fr}{\texttt{GUDHI}} library to compute persistence, and make use of the optimisation and machine learning library \texttt{PyTorch} for the construction of the graph classifications models.\\

\subsection{Results and Discussion}

In \Cref{tab:our_full_results}, we present the classification accuracies of our model. For each dataset, we perform four experiments using our model, varying whether the wavelet parameter is optimised and whether additional features are included. In \Cref{tab:binary_class_results_picnn_new}, we show the test accuracy of our model alongside two persistence-based graph classification architectures, \Perslay and \GFL, as well as other state-of-the-art graph classification architectures. 

We first compare the performances of our model between cases where we optimise and fix the wavelets. In \Cref{tab:our_full_results}, we see that on \MUTAG and \DHFR, optimising the wavelet improves the classification accuracy regardless of whether extra features are included. On \NCI, wavelet optimisation improves the classification accuracy only only persistence features are included. When we include non-persistence features in the model, the performances of the optimised and control models are statistically indistinguishable for \NCI, suggesting that the non-persistence features play a more significant role in the classification. As for \COX, \PRO, and \IMDB, optimising the wavelet coefficients do not bring about statistically significant improvements. This indicates that the initial wavelet signature \--- the heat kernel signature $W(e^{-10x})$ \--- is a locally optimal choice of wavelet for our neural network classifier.

We now compare our architecture to other persistence based architectures, \Perslay and \GFL, where node attributes are excluded from their vertex function models. Except on \PRO, our wavelet optimised model matches or exceeds \Perslay. While our model architecture and choice of wavelet initialisation is similar to that of \Perslay, we differ in two important respects. \Perslay fixes the vertex functions but optimises the weights assigned to points on the persistence diagrams, as well as the parameters of the persistence images. Our improvements on \Perslay for \MUTAG, \DHFR, and \IMDB indicate that vertex function optimisation yields improvements that cannot be obtained through vectorisation optimisation alone on some datasets of graphs. 

Compared to \GFL (without node attributes), both \Perslay and our architecture achieves similar or higher classification accuracies on \PRO and \NCI. This supports wavelet signatures being viable models for vertex functions on those datasets. On the other hand, both \Perslay and our model lag behind \GFL on \IMDB. We attribute this to the fact that \IMDB, unlike the other bionformatics datasets, consists of densely connected graphs. The graphs in \IMDB have diameter at most two and 14\% of the graphs are cliques. This fact has two consequences. First, we expect the one-layer GIN used in \GFL \--- a local topology summary \--- to be highly effective in optimising for the salient features of a graph with small diameter. Second, the extended persistence modules for cliques have zero persistence, since all vertices are assigned the same function value due to symmetry. In contrast, ordinary persistence used in \GFL is able to capture the cycles in a complete graph as points with infinite persistence. 

Compared to non-persistence state-of-the-art architectures in \Cref{tab:our_full_results}, our model achieves competitive accuracies on \MUTAG, \COX , and \DHFR. For \NCI and \PRO, all persistence architectures listed that exclude additional node attributes perform poorly in comparison, though \texttt{PWL} was able to achieve leading results with node attributes.  

All in all, we observe that wavelet signatures can be an effective parametrisation of vertex functions when we use extended persistence as features for graph classification. In particular, on some bioinformatics datasets, we show that optimising the wavelet signature can lead to improvements in classification accuracy. The wavelet signature approach is complementary to the \GFL approach to vertex function parametrisation as they show strengths on different datasets.

\section{Conclusion}

We have presented a framework for equipping any graph $G$ with a set of extended persistence diagrams $\ExtPH \circ \parametrization :\M \to \Barc^4$ parametrised over a manifold $\M$, a parameter space for the graph's wavelet signature. We described how wavelet signatures can be parametrised and interpreted. Given a function on extended persistence diagrams $\out: \Barc^4 \to \R$ that is differentiable, we have shown how a loss function $\loss = \out \circ \ExtPH \circ \parametrization$ can be generically differentiable with respect to $\theta \in \M$ as $\loss$. Thus, we can apply gradient descent methods to optimise the extended persistence diagrams of a graph to minimise $\loss$. 

We applied this framework to a graph classification architecture where the wavelet signature is optimised for classification accuracy. We are able to demonstrate an increase in accuracy on several benchmark datasets where the wavelet is optimised, and perform competitively with state-of-the-art persistence based graph classification architectures.

\section*{Funding}
Ka Man Yim is funded by the EPSRC Centre For Doctoral Training in Industrially Focused Mathematical Modelling (EP/L015803/1) with industrial sponsorship from Elsevier. Jacob Leygonie is funded by the EPSRC grant EP/R513295/1.  Both authors are members of the Centre for Topological Data Analysis, which is supported by the EPSRC grant New Approaches to Data Science: Application Driven Topological Data Analysis EP/R018472/1.

\section*{Acknowledgments}
 The authors would like to thank Ulrike Tillmann and Heather Harrington for their close guidance and thoughtful advice on this project. In addition, the authors would like to thank Vidit Nanda, Peter Grindrod CBE, Andrew Mellor, Steve Oudot, Mathieu Carrière, and Theo Lacombe for fruitful discussions on this subject. Finally, we are indebted to the reviewers for their thoughtful and constructive comments, which led to many improvements of the paper.

\section*{Data Availability Statement}
The code for the computational experiments in \cref{sec:expts} can be found in the GitHub repository \url{https://github.com/kmyim/Persistence_Opt_Spectral_Wavelets}. The datasets we use are publicly available at the repository TUDatasets \url{https://chrsmrrs.github.io/datasets/} \cite{Morris+2020}. 

\begin{table}[!h]
\centering 
\begin{tabular}{ | l | c | c | c| c ||  c | c | c | c || c  c|} \hline

 & \multicolumn{4}{c||}{\tiny Non-Persistence State-of-the-Art} & \multicolumn{6}{c|}{\tiny Persistence Based}\\ \hline 
   & {\tiny \texttt{P-SAN}} & {\tiny \texttt{RetGK}} & \tiny \texttt{GIN} & {\tiny \texttt{FGSD}} & \tiny \texttt{PWL} & \multicolumn{2}{c|}{\tiny \GFL} & \tiny{\Perslay} &  \tiny{Control} &  \tiny{Wavelet Opt.} \\ 
      & {\tiny \cite{psan}} & {\tiny \cite{retgk}} & \tiny \cite{xu2019graph} & {\tiny \cite{fgsd}} & \cite{rieck2019persistent} &  \multicolumn{2}{c |}{\tiny  \cite{hofer2020graph}} & \tiny{\cite{carriere2020perslay}} & \multicolumn{2}{c|}{\tiny  This paper}  \\ \hline 
Node attr. & \multicolumn{2}{c|}{\tiny Yes} & \multicolumn{2}{c||}{\tiny No} & \multicolumn{2}{c|}{\tiny Yes}  &  \multicolumn{4}{c|}{\tiny No} \\ \hline 
\texttt{MUTAG}   & {92.6}& \accerr{90.3}{1.1}  & 89.4 & 92.1 & \accerr{{90.5}}{1.3} & \---  & \---  & \accerr{89.8}{0.9} & 89.0\tiny{$\pm$0.6} & \textbf{90.4}\tiny{$\pm$1.3}  \\
\texttt{COX2}  & \--- & \accerr{{81.4}}{0.6}  & \--- & \---  & \--- & \---  & \---   & \accerr{\textbf{80.9}}{1.0} & 80.8\tiny{$\pm$0.4} & 80.8\tiny{$\pm$1.0}\\
\texttt{DHFR} &  \--- &  \accerr{{82.5}}{0.8}  & \--- & \--- & \--- & \---  & \--- & \accerr{80.3}{0.8} & 80.0\tiny{$\pm$0.4} & {\textbf{81.0}}\tiny{$\pm$0.9}\\
\texttt{NCI1} & 78.6 & \accerr{84.5}{0.2} & 82.7 & 79.8 & \accerr{{85.6}}{0.3} & 77.2 & 71.2 & \accerr{73.5}{0.3}  &  74.3\tiny{$\pm$0.3}  & \textbf{74.4}\tiny{$\pm$0.3}   \\
\texttt{PROTEINS} & 75.9 & \accerr{{78.0}}{0.3} & 76.2 & 73.4 & \accerr{{75.9}}{0.8} & 73.4 & 74.1 &  \accerr{\textbf{74.8}}{0.3} &  74.5\tiny{$\pm$0.4}  &  74.6\tiny{$\pm$0.6} \\
\texttt{IMDB-B} & 71.0 & \accerr{72.3}{0.6} & {75.1} & 73.6 & \accerr{73.0}{1.0} & \--- & \textbf{74.5} & \accerr{71.2}{0.7} &  71.6\tiny{$\pm$0.9}  &  72.0\tiny{$\pm$0.7}   \\\hline 
\# Ten-folds & 10 & 10 & 1 & 1 & 10 &1 & 1 & 10 &  10  & 10  \\\hline 
\end{tabular}
\bigskip
\caption{Binary classification accuracy on datasets of graphs. The best accuracy of persistence-based models \emph{without} using node attributes is made bold for each dataset. The performance of our model is reported in the column Wavelet Opt. on the right hand side. The accuracies of the control model, where the wavelet parameters are fixed to the initial values, are shown in the column Control. Both these models use additional features (see \Cref{sec:non-pers}). The accuracies of our model are the means over 10 ten-folds, recorded at epochs reported in \Cref{tab:our_full_results_epochs}. We also provide the standard deviations of the 10 mean accuracies of each ten-fold. For other architectures, we indicate whether their accuracies were reported as averages over 1 ten-fold or 10 ten-fold in the bottom row of the table. To avoid confusion, we leave out the errors reported for \texttt{P-SAN}, \texttt{GIN} and \GFL and refer the reader to the original sources, as they were calculated using a different formula. Errors were not reported in \cite{fgsd} for \texttt{FGSD}.  \label{tab:binary_class_results_picnn_new}}
\end{table}

\begin{table}[!h]
\centering
\begin{tabular}{ | l | c  c  | c  c |} \hline

 & \multicolumn{2}{c|}{\tiny Persistence Only} & \multicolumn{2}{c|}{\tiny Non-Persistence Features incl.}\\ \hline 
 &\tiny Control & \tiny Wavelet Opt. & \tiny Control & \tiny Wavelet Opt. \\ \hline 

\texttt{MUTAG} & \accerr{89.2}{0.6}  &\accerr{89.8}{0.8} & \accerr{89.0}{0.6} &\accerr{90.4}{0.4} \\
\texttt{COX2} & \accerr{79.6}{1.0} & \accerr{79.4}{0.7} &\accerr{80.8}{1.0} & \accerr{80.8}{1.0}  \\
\texttt{DHFR} &\accerr{79.9}{0.4} & \accerr{80.4}{0.4} &\accerr{80.3}{0.9} & \accerr{81.0}{0.9}  \\
\texttt{NCI1} &\accerr{73.7}{0.2} & \accerr{74.3}{0.5} &\accerr{74.3}{0.3} & \accerr{74.4}{0.3}  \\
\texttt{PROTEINS} &\accerr{72.9}{0.3} & \accerr{73.0}{0.4} &\accerr{74.5}{0.4} & \accerr{74.6}{0.6}  \\
\texttt{IMDB-B} &\accerr{68.3}{0.5} & \accerr{68.6}{0.7} &\accerr{71.6}{0.9} & \accerr{72.0}{0.7}  \\ \hline 
\end{tabular}
\bigskip
\caption{Binary classification accuracy of our model where we vary whether non-Persistence features are included and whether the wavelet is optimised. The reported accuracies are the mean over 10 ten-folds, recorded at epochs reported in \Cref{tab:our_full_results_epochs}. We also provide standard deviations of the 10 mean accuracies of each ten-fold. See \Cref{sec:non-pers} for the particulars about the non-persistence features. \label{tab:our_full_results}}
\end{table}

\begin{table}[!h]
\centering 
\begin{tabular}{ | l | c  c  | c  c |} \hline

 & \multicolumn{2}{c|}{\tiny Persistence Only} & \multicolumn{2}{c|}{\tiny Non-Persistencee Features incl.}\\ \hline 
  & \tiny Control & \tiny Wavelet Opt. &\tiny Control & \tiny Wavelet Opt. \\ \hline 

\texttt{MUTAG} & 25 & 125 & 25 & 75 \\
\texttt{COX2} & 50 & 50 & 25 & 25 \\
\texttt{DHFR} & 125 & 250 & 125 & 45  \\
\texttt{NCI1} & 270 & 270 & 500 &  370 \\
\texttt{PROTEINS} & 50 & 50 & 125 & 125 \\
\texttt{IMDB-B} & 100  & 25  & 75 & 50 \\ \hline 
\end{tabular}
\bigskip
\caption{Epochs at which accuracies in \Cref{tab:our_full_results} are recorded.  \label{tab:our_full_results_epochs}}
\end{table}

\begin{table}[h!]
\centering
\begin{tabular}{| l| lllll|}
\hline
           & \MUTAG & \COX & \DHFR & \NCI & \IMDB \\ \hline
\# graphs  & 188   & 467  & 756  & 4110 & 1000 \\
batch size & 10    & 9   & 11   & 20  & 50 \\ \hline 
\end{tabular}
\bigskip
\caption{Batch sizes in the graph classification experiments for different datasets described in \cref{sec:expts}. \label{tab:batchsize}}
\end{table}
\newpage
\newpage
\printbibliography
\newpage
\begin{appendices}

\section{Differentiability of the extended persistence map}\label{app:diff}

Let $K$ be a finite simplicial complex with vertex set~$V$ and dimension~$d\in\N$. A vertex function~$f\in \R^V$ extends to the whole complex via~$f(\simplex):=\max_{v\in \simplex} f(v)$. Filtrations, persistence modules and barcodes are then defined analogously to the case of a graph. The extended barcode of a function~$f$ now consists of~$3(d+1)$ barcodes:
\begin{equation}
    \ExtPH[][f][] = \qty[\{\ExtPH[k][f][p]\}_{k\in \{\mathrm{ord},\mathrm{ext},\mathrm{rel}\}}]_{p=0}^d \in \Barc^{3(d+1)}.
\end{equation}

We then have the extended persistence map 
\[ \ExtPH[][][]: \R^V \to \Barc^{3(d+1)},\]
and the ordinary persistence map as in remark~\ref{remark_ordinary_persistence}
\[\OrdPH: f\in \R^K \longmapsto \qty[\OrdPH_p(f)]_{p=0}^d  \in \Barc^{d+1}.\]

\begin{proposition}
\label{proposition_extended_generically_differentiable_appendix}
Let~$K$ be a finite simplicial complex, and let $\parametrization: \M \rightarrow \R^V$ be a generically differentiable parametrisation. Then the composition $\ExtPH[][][] \circ \parametrization$ is generically differentiable.
\end{proposition}

In particular, taking the parameter space~$\M$ to be the space~$\R^V$ of vertex functions, we obtain the generic differentiablility of the extended persistence map~$\ExtPH[][][]$ itself. Note that, however, we could not have directly deduced the generic differentiability of any composition of the form~$\ExtPH[][][] \circ \parametrization$ from the generic differentiability of~$\ExtPH[][][]$. This is due to the fact that the image of a parametrisation~$\parametrization$ might lie in the set where~$\ExtPH[][][]$ is not differentiable. 

The idea of our proof is to view the extended persistence of a vertex function~$f\in \R^V$ as the ordinary persistence of an extension of~$f$ over the cone complex~$\Cone[K]$. We note that this point of view has proven to be particularly useful for computing extended persistence in practice~\cite{cohen2009extending}. The relationship between~$\ExtPH[][][]$ and~$\OrdPH[]$ can be described by a commutative diagram:

\begin{equation*}
\label{eq_diagram_connection_ordinary_extended_no_maps}
\begin{tikzpicture}

\node[] (S) at (-2,2) {$\R^V$};
\node[] (D) at (-2,0) { $\R^{\Cone[K]}$};

\node[] (DeltaS) at (2,2) {$\Barc^{3}$};
\node[] (DeltaD) at (2,0) {$\Barc$};

\draw[->] (S)--(D);
\draw[->] (DeltaD)--(DeltaS);
\draw[->] (S)--(DeltaS) node[midway, above]{$\ExtPH$};
\draw[->] (D)--(DeltaD) node[midway, above]{$\OrdPH$};
\end{tikzpicture}
\end{equation*}
whose vertical maps are differentiable. Thus, we can deduce the differentiability of the extended persistence map~$\ExtPH$ from the results of~\cite{leygonie2019framework} about the ordinary persistence map~$\OrdPH$.

\begin{proof}[Proof of Proposition~\ref{proposition_extended_generically_differentiable_appendix}]

Let $f\in \R^V$ be a vertex function. Let~$K^t$ (resp. $K_t$) be the maximal sub complexes of~$K$ induced by vertices taking values greater (resp less) than $t$. For $0\leqslant p\leqslant d$, the associated $p$-th extended persistent homology module~$\persmod_p(f)$ is:

\begin{equation}
\label{eq_extended_persistent_homology_complex}
\begin{tikzcd}[column sep=small]
0 = H_p(\emptyset) \arrow[r] & \cdots \arrow[r] & H_p(K_s) \arrow[r, "s \leq t"] & H_p(K_t) \arrow[r] & \cdots \arrow[r]  & H_p(K) \arrow[d, "\cong"] \\
0 = H_p(K,K) & \arrow[l]  \cdots & \arrow[l] H_p(K,K^s)  & \arrow[l, "s \leq t"'] H_p(K,K^t)  & \arrow[l] \cdots  & \arrow[l]  H_p(K,\emptyset)
\end{tikzcd}.
\end{equation}
As such,~$\persmod_p(f)$ is a module indexed over the extended real line~$\R \sqcup \{\infty\}\sqcup \R^{\mathrm{op}}$. We construct an equivalent module~$\persmod_{p,R}(f)$ over the simpler, compact poset~$[-R;3R]$, where~$R>0$ is a large enough constant chosen hereafter. For this, we consider the poset map that collapses~$\R \sqcup \{\infty\}\sqcup \R^{\mathrm{op}}$ onto $[-R;3R]$ as in fig.~\ref{fig:posetmap}. Formally, the poset map is defined on~$\R$ as the canonical retraction onto~$[-R;R]$, on~$\R^{\mathrm{op}}$ as the symmetry $t\mapsto 2R-t$ followed by the canonical retraction onto~$[-R;R]$, and sends the point~$\infty$ to~$R$.
\begin{figure}[h!]
    \centering
     \includegraphics[width=\textwidth]{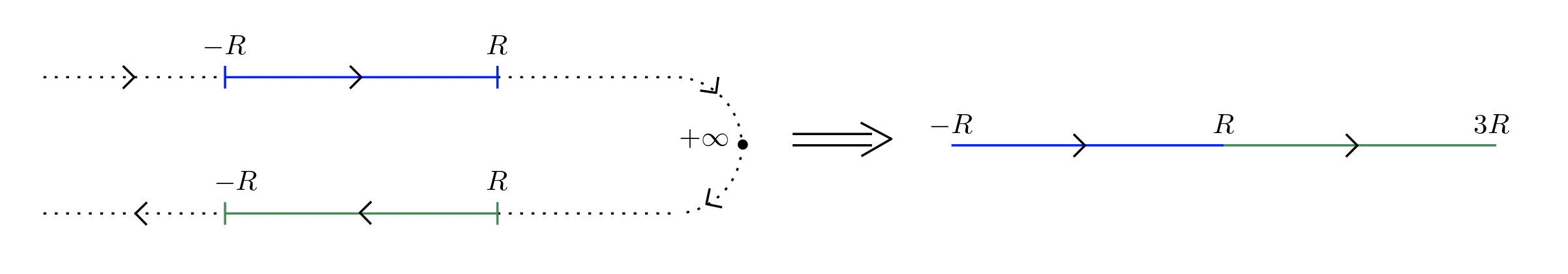}
    \caption{Collapsing the dotted part of the left poset yields the compact poset on the right.}
    \label{fig:posetmap}
\end{figure}

If we choose $R>\sup_{\simplex}|f(\simplex)|$, then no new simplex enters the subcomplexes~$K_t$ or~$K^t$ for~$t\notin [R;R]$ and~$t\notin [R;R]^{\mathrm{op}}$. Hence, the module~$\persmod_p(f)$ is locally constant outside of~$[R;R]$ and~$[R;R]^{\mathrm{op}}$. By applying the (inverse of the) poset map described above, we thus get a $[-R;3R]$-indexed module~$\persmod_{p,R}(f)$:
\begin{equation}
\label{eq_extended_persistent_homology_shrink}
\begin{tikzcd}[column sep=small]
0 = H_p(K_{-R}) \arrow[r] & \cdots \arrow[r] & H_p(K_s) \arrow[r, " s \leq t "] & H_p(K_t) \arrow[r] & \cdots \arrow[r]  & H_p(K_R) \arrow[d, "\cong"] \\
0 = H_p(K,K^{-R}) & \arrow[l]  \cdots & \arrow[l] H_p(K,K^{2R-t})  & \arrow[l, " t \geq s "'] H_p(K, K^{2R-s})  & \arrow[l] \cdots  & \arrow[l]  H_p(K;K^R)
\end{tikzcd}.
\end{equation}

The extended module module~$\persmod_{p}(f)$ is essentially equivalent to the ordinary module~$\persmod_{p,R}(f)$, since we can retrieve the extended barcode~$\ExtPH[][f][p]$ of~$\persmod_{p}(f)$ from the barcode of~$\persmod_{p,R}(f)$ as follows:

\begin{itemize}
    \item Each interval $\expval{b,d}$ in the barcode of~$\persmod_{p,R}(f)$ such that $b\leqslant d <R$ yields an interval $\expval{b,d}$ in $\ExtPH[ord][f][p]$;
    \item Each interval $\expval{b,d}$ in the barcode of~$\persmod_{p,R}(f)$ such that $b<R<d$ yields an interval $\expval{b,2R-d}$ in $\ExtPH[ext][f][p]$;
    \item Each interval $\expval{b,d}$ in the barcode of~$\persmod_{p,R}(f)$ such that $R<b\leqslant d$ yields an interval $\expval{2R-b,2R-d}$ in $\ExtPH[rel][f][p]$.
\end{itemize}

We denote this decoding map by $\mathrm{Dec}_R: \Barc \rightarrow \Barc^{3}$. We next take advantage of working with the ordinary module~$\persmod_{p,R}(f)$ by viewing it as the sub level set persistent homology module of a function defined on the cone~$\Cone[K]$.  

Note that the relative homology groups of~$\persmod_{p,R}(f)$ in the second row of Eq.~\eqref{eq_extended_persistent_homology_shrink} may be replaced with ordinary (reduced) homology groups of the cones $\Cone[K^{2R-t}]$ using the functorial isomorphism:
\[ H_p(K,K^{2R-t})\cong \tilde{H}_p(K/K^{2R-t})\cong \tilde{H}_p(K\cup \Cone[K^{2R-t}]).\] 
We denote by~$\omega$ the distinguished vertex of such cones. It is then clear that~$\persmod_{p,R}(f)$ equals the ordinary $p$-th sub level set persistent (reduced) homology module of the function $\hat{f}_R: \Cone[K]\rightarrow \R$ defined by 
\[\hat{f}_R(\simplex):=f(\simplex) \text{     and        } \hat{f}_R(\simplex\sqcup \{\omega\}):=2R-\min_{v \text{ vertex in } \simplex} f(v)\]
for any simplex $\simplex\in K$, and $\hat{f}_R(\omega):=-R$ by convention. Plugging these constructions together, we connect the ordinary and extended maps in the commmutative diagram:\footnote{Strictly speaking, the decoding map should furthermore forget the unique unbounded interval~$\expval{b,+\infty}$ in the barcode~$\OrdPH[\hat{f}_R]$, since the ordinary persistence map~$\OrdPH[]$ computes the barcode of a module made of non-reduced homology groups.} 
\begin{equation}
\label{eq_diagram_connection_ordinary_extended}
\begin{tikzpicture}

\node[] (S) at (-2,3) {$\R^V$};
\node[] (D) at (-2,0) { $\R^{\Cone[K]}$};

\node[] (DeltaS) at (2,3) {$\Barc^{3}$};
\node[] (DeltaD) at (2,0) {$\Barc$};

\draw[->] (S)--(D);

\draw[->] (DeltaD)--(DeltaS) node[midway, right] {$\mathrm{Dec}_R$};
\draw[->] (S)--(DeltaS) node[midway, above]{$\ExtPH$};
\draw[->] (D)--(DeltaD) node[midway, above]{$\OrdPH$};
\node[] (f) at (-1.5,2.7) {$f$};
\node[] (hatf) at (-1.5,0.6) {$\hat{f}_R$};
\draw[|->] (f)--(hatf);
\node[] (M) at (-6,1.5) {$\M$};
\draw[->] (M)--(S) node[midway, above] {$\parametrization$};
\draw[->, dotted] (M)--(D); 
\node[] (theta) at (-5.8,1) {$\theta$};
\node[] (FRtheta) at (-3.6,0.15) {$\hat{\parametrization(\theta)}_R$};
\draw[|->] (theta)--(FRtheta);
\end{tikzpicture}
\end{equation}
Note that this diagram only makes sense for parameters~$\param$ such that~$\hat{\parametrization(\param)}_R$ is a function whose sub level sets are sub complexes of~$\Cone[K]$, as~$\OrdPH[\hat{\parametrization(\param)}_R]$ is undefined otherwise. This requirement is satisfied whenever the inequality $\sup_{\simplex} |\parametrization(\param)(\simplex)|<R$ holds. For simplicity, we assume that~$R$ can be chosen large enough for the inequality to hold for all parameters~$\param$, hence the diagram~\eqref{eq_diagram_connection_ordinary_extended} makes sense globally on~$\M$. One can always avoid this restriction by working locally on compact neighbourhoods in~$\M$.  

From~\cite[Theorem~4.9]{leygonie2019framework}, the subset~$\M'\subseteq \M$ where the parametrisation~$\parametrization$ is differentiable and induces a locally constant pre-order on simplices of~$K$ is a generic sub manifold. In turn, all the maps $\param \mapsto \min_{v\in \simplex} \parametrization(\param)(v)$ and $\param \mapsto \max_{v\in \simplex} \parametrization(\param)(v)$, for $\simplex\in K$ a simplex, are differentiable over~$\M'$. Therefore $\hat{\parametrization}_R:\M \rightarrow \R^{\Cone[K]}$ is differentiable over the generic submanifold~$\M'$.

Since~$\hat{\parametrization}_R$ is generically differentiable, so is $\OrdPH \circ \hat{\parametrization}_R$~\cite[Theorem~4.9]{leygonie2019framework}, i.e. we generically have local coordinate systems as in Def.~\ref{definition_differentiability} . Since the decoding map $\mathrm{Dec}_R$ in diagram~\eqref{eq_diagram_connection_ordinary_extended} merely applies an affine transformation to the local coordinate systems and then splits them into three parts (the splitting is constant), we obtain local coordinate systems for~$\ExtPH\circ F$. Therefore, $\ExtPH\circ F$ is generically differentiable.
\end{proof}

\section{The Wavelet Signature is Well-defined} \label{app:wav}
In \cref{def:wavelet}, we defined the wavelet signature using the eigenvalues and eigenvectors of a graph Laplacian $L$. The wavelet signature is only well defined if it is independent of the choice of eigenbasis for $L$, where ambiguity could occur if $L$ has eigenvalues with multiplicity\footnote{As $L$ is symmetric and hence diagonalisable, the geometric and algebraic multiplicities of its eigenvalues agree.} greater than one. 
\waveletprop*

\begin{proof} Let $\mathsf{Spec}(L) \subset \R$ denote the spectrum of $L$ and  $\vb*{\phi}_1, \ldots \vb*{\phi}_{\qty|V|}$ be a set of orthonormal eigenvectors of $L$. Let us denote $\Phi(\lambda)$ to be a $\qty|V| \times m$ matrix where $m$ corresponds to the geometric multiplicity of $\lambda$, and the $m$ column vectors of $\Phi(\lambda)$ correspond to eigenvectors $\vb*{\phi}_{i_1},\ldots, \vb*{\phi}_{i_m}$ with eigenvalue $\lambda$. Then we can rewrite the wavelet signature \cref{eq:WSdef} as  
\begin{equation} \label{eq:WSrewritten}
    W(g)_v = \sum_{i=1}^{\qty|V|} g(\lambda_i) \qty(\vb*{\phi}_{i})_v^2 = \sum_{\lambda \in \mathsf{Spec}(L)} g(\lambda) \qty(\Phi(\lambda) \Phi(\lambda)^\intercal)_{vv}
\end{equation}
Suppose we have another choice of eigenbasis of $L$. Without loss of generality for $\lambda \in \mathsf{Spec}(L)$, the new basis $\vb*{\phi}'_{i_1},\ldots, \vb*{\phi}'_{i_m}$ for $\mathsf{eig}(\lambda)$ is related to the previous eigenbasis $\vb*{\phi}_{i_1},\ldots, \vb*{\phi}_{i_m}$ by an orthonormal transformation transformation $U(\lambda) \in \R^{m \times m}$ on $\Phi(\lambda)$:
\begin{equation*}
   \Phi'(\lambda) =  \qty[\ \vb*{\phi}'_1\  \cdots \ \vb*{\phi}'_m\ ]   = \Phi(\lambda) U(\lambda).
\end{equation*}
As $U(\lambda)$ is an orthonormal transformation with $U(\lambda) U(\lambda)^\intercal = 1$, 
\begin{align*}
    \Phi'(\lambda) \Phi'(\lambda)^\intercal & = \Phi(\lambda) U(\lambda)  \qty(\Phi(\lambda) U(\lambda))^\intercal \\
    &= \Phi(\lambda) U(\lambda) U(\lambda)^\intercal \Phi(\lambda)^\intercal \\
    &= \Phi(\lambda) \Phi(\lambda)^\intercal.
\end{align*}
Since the $V \times V$ matrix $\Phi(\lambda) \Phi(\lambda)^\intercal$ is independent of the choice of eigenbasis, the wavelet signature given on the right hand side of \cref{eq:WSrewritten} must also be independent of the choice of eigenbasis.
\end{proof}
\newpage 
\begin{figure}[!h]
    \centering
    \includegraphics[width = \textwidth]{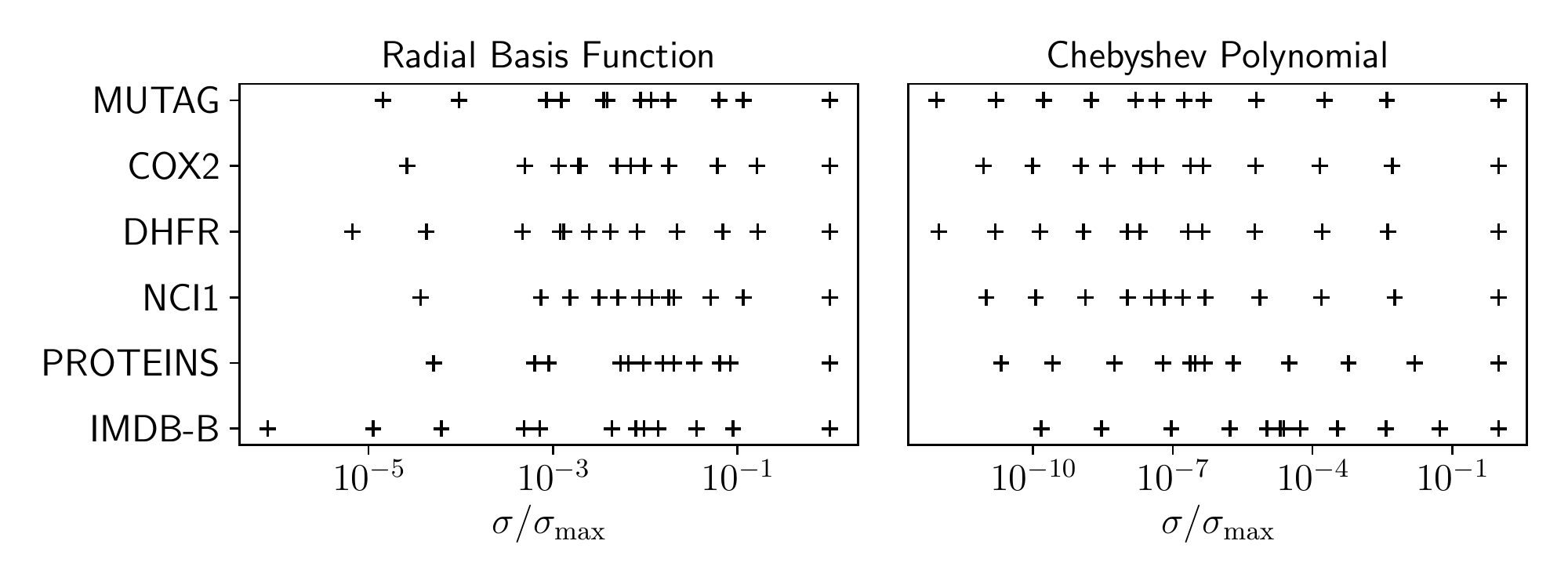}
    \caption{We consider the parametrisations of wavelet signatures on some datasets of graphs in machine learning, namely \MUTAG, \COX, \DHFR, \NCI, \PRO and \IMDB, using coefficients of 12 radial basis functions (see \cref{eq:rbf_actual}) and a degree 13 Chebyshev polynomial respectively. For each dataset, we plot the distribution of the singular values $\sigma$ of the map $\parametrization$ in \cref{eq:param_whole_dataset} from the basis function coefficients $\theta \in \R^{12}$ to the wavelet signature on the whole dataset of graphs, as a fraction of the largest singular value $\sigma_\mathrm{max}$ of $\parametrization$. We can observe that for both parametrisations, the singular values span many orders of magnitudes across different datasets. Note that the singular values of $\parametrization$ not only depend on the choice of basis but also on the dataset of graphs.}
    \label{fig:singular_values}
\end{figure} 

\begin{figure}[!h]
    \centering
    \includegraphics[width = \textwidth]{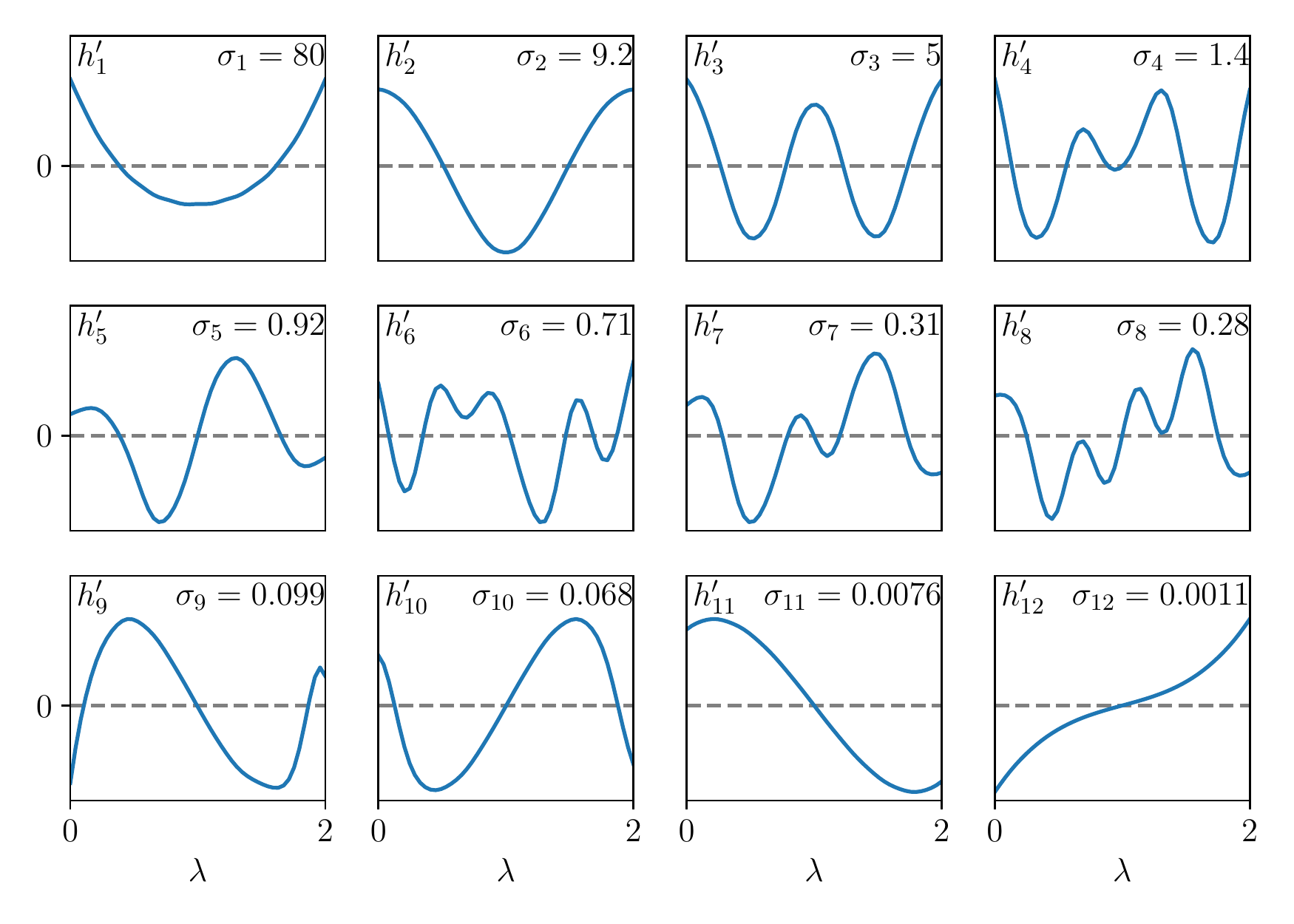}
    \caption{The functions shown are the new, stable wavelet basis $h_1', \ldots , h_{12}'$ (\cref{eq:stable_basis}) for the \MUTAG dataset, derived from an initial numerically unstable parametrisation using twelve inverse multiquadric radial basis functions (\cref{eq:rbf_actual}). We parametrise the wavelet as a linear combination of these basis functions.}
    \label{fig:mutag_basis}
\end{figure}

\section{Experimental Details} \label{app:expt}

\subsection{Persistence Images Parameters}

We vectorised each of the three persistence diagrams $\ExtPH[][][0],\ \ExtPH[rel][][1]$ and $\ExtPH[ext][][1]$ as a persistence image. Prior to vectorising the persistence diagrams, we apply a fixed and identical affine transformation to the values of the vertex functions across all graphs in the dataset concerned, such that the maximum and minimum values taken across all graphs in the dataset of the initial vertex function prior to optimisation are scaled to 1 and 0 respectively. The persistence image is sampled on a $20 \times 20$ grid, whose grid points are equidistantly placed $\sigma = 1/17$ apart on the square $[-\sigma,1+\sigma]^2$ of the persistence diagrams, where $\sigma$ is the width of the Gaussian. The Gaussian centred on the birth and persistence coordinates $\expval{b,p}$ of each point is weighted according to its persistence
\begin{equation*}
    \omega(p) = \sin[2](\frac{\pi}{2}\min\qty(\frac{p}{\sigma}, 1)).
\end{equation*}
Points with persistence $p \geq \sigma$ are assigned a uniform weight $\omega = 1$, else assigned a weight that diminishes to zero as $p \to 0$.  \\ 

\subsection{Convolutional Neural Network Architecture for Persistence Images}
We feed each set of three persistence images belonging to either the optimisable or static persistence diagrams as a three channel image into the following convolutional neural network to obtain a 22$\times$22 image:
\begin{equation} 
    \begin{tikzcd}[column sep = huge ]
    \mathsf{CNN} : \R^{3 \times 20\times 20}  \arrow[rr, "{\mathsf{ReLU} \ \circ \ \mathsf{Conv}\ \circ\ \mathsf{BN2D}}"] && \R^{20 \times 21 \times 21} \arrow[rr, "{\mathsf{ReLU} \ \circ \ \mathsf{Conv}\ \circ\ \mathsf{DO}\ \circ \  \mathsf{BN2D}}"]  &&\R^{22 \times 22}.
    \end{tikzcd} 
\end{equation}
The function $\mathsf{Conv}$ denotes a convolutional layer with kernel size 2, stride 1, padding 1; $\mathsf{BN2D}$ denotes a 2D batch normalisation layer; and $\mathsf{DO}$ denotes a dropout layer with dropout probability $p = 0.5$. \\

\subsection{Multilayer Perceptron for Non-Persistence Features}
We feed non-persistence features as a vector of length $n = \# \text{features}$ into the following multilayer perceptron:
\begin{equation} 
    \begin{tikzcd}[column sep = huge ]
    \mathsf{MLP} : \R^n  \arrow[rr, "{\mathsf{BN} \ \circ \ \mathsf{ReLU} \ \circ \ \mathsf{Aff}\ \circ\ \mathsf{BN}}"] && \R^{n} 
    \end{tikzcd} 
\end{equation}
where $\mathsf{Aff}: \R^n \to \R^n$ denotes an affine transformation, and $\mathsf{BN}$ denotes a batch normalisation layer. \\

\subsection{Path Encoding of Laplacian Eigenvalues}\label{app:expt_path}
For \MUTAG, \COX, \DHFR, and \NCI, we sort the Laplacian eigenvalues in ascending order and transform the one-dimensional sequence into a two-dimensional time series via a \emph{delay embedding}
\begin{equation}
    \mqty[\lambda_1 &  \lambda_2 & \cdots  &  \lambda_{N-1} & \lambda_N ] \mapsto \mqty[ \mqty[\lambda_1 \\ \lambda_2] & \mqty[\lambda_2 \\ \lambda_3] & \cdots &  \mqty[\lambda_{N-2} \\ \lambda_{N-1}] & \mqty[\lambda_{N-1} \\ \lambda_{N}]].
\end{equation}
For \IMDB, we incorporate a fictitious time coordinate $t_j = 2(j-1)/(N-1)$ for $j = 1, \ldots N$ as the second coordinate instead of a `delayed' eigenvalue:
\begin{equation}
    \mqty[\lambda_1 &  \lambda_2 & \cdots  &  \lambda_{N-1} & \lambda_N ] \mapsto \mqty[ \mqty[\lambda_1 \\ t_1] & \mqty[\lambda_2 \\ t_2] & \cdots &  \mqty[\lambda_{N-1} \\ t_{N-1}] & \mqty[\lambda_{N} \\ t_{N}]].
\end{equation}


\end{appendices}
\end{document}